\documentclass[11pt]{article}
\usepackage[]{acl}
\usepackage{times}
\usepackage{latexsym}
\usepackage{amsmath}
\usepackage{amssymb}
\usepackage{booktabs}
\usepackage{graphicx}
\usepackage{multirow}
\usepackage{xcolor}
\usepackage{subcaption}
\usepackage{float}
\usepackage{amsthm}
\usepackage{hyperref}
\newtheorem{theorem}{Theorem}
\newtheorem{proposition}{Proposition}

\hypersetup{
  pdftitle={RegimeAbstain: Retrieval Confidence Scoring and Abstention for Multi-Hop QA},
  pdfauthor={Andre Bacellar},
  pdfkeywords={retrieval abstention, retrieval confidence, CWAR, selective prediction,
               retrieval-augmented generation, uncertainty quantification},
}

\newcommand{\system}{\textsc{RegimeAbstain}}
\newcommand{\router}{\textsc{RegimeRouter}}
\newcommand{\bridgerag}{\textsc{BridgeRAG}}

\newcommand{\rcs}{\mathrm{RCS}}
\newcommand{\abstain}{\bot}
\newcommand{\thresh}{\tau}

\title{Predictable Failure in Multi-Hop Retrieval:\\
       Score-Distributional Confidence Scoring and Abstention}

\author{Andre Bacellar \\
  \texttt{andremi@gmail.com}}

\begin{document}
\maketitle

\begin{abstract}
Multi-hop retrieval failures are not uniformly distributed across queries:
they cluster in structurally predictable subpopulations.
We prove two results formalizing this structure.
First (\emph{CWAR Reducibility}, Theorem~\ref{thm:reducibility}): confident-failure
reduction is achievable if and only if retrieval features carry mutual information
about success — a condition satisfied by LLM-judge pipelines but substantially
weaker in dense-only settings, explaining the AUC-AC gap between regimes.
Second (\emph{Feature Regime Complementarity}, Proposition~\ref{thm:complementarity}):
no single ANN score feature achieves best predictive performance across all failure
regimes; the dominant feature differs between datasets (query length on MuSiQue,
hop-1 concentration on HoVer), and a constructive witness pair shows each is
necessary in one regime and non-contributory in the other.
We instantiate these principles in \system{}, which computes a
\emph{Retrieval Confidence Score} ($\rcs$) — a logistic function of up to nine
query-ANN structural features, all available without any additional LLM call —
and uses it to implement a calibrated abstention policy.
We define the \emph{Confident-Wrong-Answer Rate} (CWAR) metric and evaluate across
three multi-hop benchmarks (MuSiQue, 2WikiMultiHopQA, HoVer) and two retrieval
architectures (LLM-judge and dense-only), covering five failure regimes with CWAR
from 14.5\% to 62.1\%.
$\rcs$ achieves best or co-best AUC-AC in all five conditions against eight
confidence baselines.
On MuSiQue (LLM-judge), $\rcs$ reduces CWAR from 39.5\% to 20.6\% at 50\%
coverage (47.8\% relative reduction), with ECE\,$=$\,0.035.
A model trained on MuSiQue transfers to 2WikiMultiHopQA with only $-0.5$pp AUC
loss, confirming the domain-agnostic structure of regime features.
\end{abstract}

\section{Introduction}
\label{sec:intro}

Retrieval-augmented generation (RAG) pipelines always produce an answer: given a
query, the system returns its best-matching top-$k$ passages and feeds them to
a language model.
This design is sensible when the goal is to maximize recall across a broad query
distribution.
It is inappropriate when retrieval quality is heterogeneous and high-confidence
errors are costly.
Multi-hop QA is precisely such a setting: a query requires retrieving both a bridge
passage and an answer passage, and the system frequently succeeds on one while
missing the other.
On MuSiQue with our LLM-judge pipeline \citep{bacellar2026bridgerag},
39.5\% of test queries fail to retrieve \emph{all} gold passages into the top-5.
These are not low-confidence retrievals—the system returns a ranked list regardless,
with no signal to downstream components that the evidence is incomplete.

We argue that a retrieval system operating on heterogeneous queries should support
a third action: \emph{abstention} ($\abstain$).
When the system predicts low retrieval quality, it can flag the query rather than
silently returning an incomplete context.
Abstracting this into a calibrated score opens the door to query routing,
graceful degradation, and human-in-the-loop escalation—all without modifying
the underlying retrieval pipeline.

We study what makes a query hard at retrieval time, without using gold labels.
The insight comes from the regime-conditional framework of
\citet{bacellar2026regime}: queries that sit near the transition between
Q-dominant and B-dominant retrieval regimes are structurally ambiguous.
Lightweight features from the ANN score distribution prove particularly predictive:
\emph{hop-1 lift} (how strongly the top ANN passage dominates the rest),
\emph{hop-2 max score} (quality of the expanded query search), and
\emph{query length} (longer questions correlate with harder multi-hop chains).

\paragraph{Contributions.}
\begin{enumerate}
  \item We prove that CWAR is reducible if and only if retrieval features carry
        mutual information about retrieval success, formally characterizing when
        abstention can help (Theorem~\ref{thm:reducibility}).
  \item We show empirically that no single ANN score feature achieves best
        predictive performance across all failure regimes
        (Proposition~\ref{thm:complementarity}): a constructive witness pair
        exhibits opposite necessity/redundancy patterns across MuSiQue and HoVer,
        motivating multi-feature aggregation.
  \item We define the \emph{Confident-Wrong-Answer Rate} (CWAR) metric and
        introduce $\rcs$, a logistic function of nine query-time features
        that is calibrated (ECE $= 0.035$) and requires no additional LLM inference.
  \item We evaluate $\rcs$ against eight confidence baselines across five failure
        regimes spanning three datasets and two retrieval architectures,
        establishing consistent best or co-best AUC-AC in all conditions.
  \item We show cross-dataset transfer ($-0.5$pp AUC) and identify regime-specific
        dominant features: query length on MuSiQue, hop-1 concentration on HoVer.
  \item We show that choosing the abstention threshold by its tune-set CWAR misses
        the CWAR target on new queries in most re-splits (56--67\%), and that a
        guaranteed-risk selection meets it with the stated confidence, at a measured
        cost in coverage (Section~\ref{sec:threshold-reliability}).
\end{enumerate}

\section{Background and Related Work}
\label{sec:related}

\paragraph{Selective prediction.}
The study of abstaining classifiers dates to \citet{chow1970optimum}.
\citet{geifman2017selective} formalize the accuracy-coverage trade-off: at
coverage $c$, the system answers only the fraction $c$ of queries it is most
confident about, and accuracy among answered queries should increase as $c$
decreases.
\citet{kamath2020selective} apply this framework to open-domain QA, training a
confidence predictor on top of a QA model's internal states.
Our setup differs: we predict \emph{retrieval} success (do the retrieved passages
contain the gold answers?), not answer generation accuracy, using only pre-LLM
features.

\paragraph{LLM uncertainty.}
Post-generation uncertainty is well-studied: \citet{kadavath2022language}
show that large LLMs can self-predict answer accuracy; \citet{kuhn2023semantic}
propose semantic entropy over sampled outputs.
These require a full generation pass per query.
$\rcs$ is computed in $<1$ms from ANN scores available at retrieval time---
no LLM call is needed for the abstention decision.

\paragraph{Conformal prediction for NLP.}
Conformal prediction \citep{angelopoulos2022conformal} provides coverage
guarantees under exchangeability.
We include a split-conformal baseline (Section~\ref{sec:setup}) and show that
$\rcs$ outperforms it on the AUC-AC metric while offering a continuous score
for operating-point selection.

\paragraph{Confidence-gated retrieval.}
Several systems act on a retrieval- or query-level confidence signal.
CRAG scores retrieved documents and triggers corrective actions such as web
search \citep{yan2024crag}; Adaptive-RAG routes each query to no, single-step
or multi-step retrieval by its predicted complexity \citep{jeong2024adaptive};
TRAQ uses conformal prediction for an end-to-end correctness guarantee on
retrieval-augmented answers \citep{li2024traq}.
\citet{li2026confidence} survey such uses of confidence as a control signal.
$\rcs$ differs in what it predicts and when: the probability that retrieval
itself succeeded, computed from score distributions before any answer is
generated, with abstention as the action.

\paragraph{Score-distributional signals for retrieval abstention.}
Concurrent work by \citet{holdcroft2026magnitude} shows, on single-hop
semantic, logical, and temporal benchmarks, that thresholding raw similarity
magnitude is a poor abstention signal and that zero-cost distributional
statistics of the retrieved score list (score gap, magnitude-and-variance)
improve abstention AUROC by up to 0.16.
We reach the same conclusion from a different direction and extend it in four
ways: the setting is multi-hop, where failure concentrates in structurally
identifiable regimes; the signals are combined by a learned model whose output
is a calibrated probability (ECE $=0.035$) rather than used as raw rankers; the
analysis covers an LLM-judge pipeline as well as dense retrieval; and
Theorem~\ref{thm:reducibility} gives the condition under which any such signal
can reduce confident failure at all.

\paragraph{Regime-conditional retrieval.}
\citet{bacellar2026regime} show that multi-hop queries divide into two regimes:
Q-dominant (bridge entity contained in the question) and B-dominant (bridge
entity must be discovered from the hop-1 result).
\system{} uses related features for a different task: predicting retrieval
success rather than choosing a pipeline, providing a complementary abstention
gate for the same regime-aware stack.

\section{Problem Formulation}
\label{sec:problem}

\subsection{Confident-Wrong-Answer Rate (CWAR)}

Let $\mathcal{Q}$ be a query set, $R: q \to \{p_1, \ldots, p_k\}$ a retrieval
system, and $\mathcal{G}(q)$ the set of gold passages for query $q$.
Define full success:
\[
  y(q) = \mathbf{1}\!\left[\mathcal{G}(q) \subseteq \{p_1,\ldots,p_k\}\right]
\]
i.e., 1 if all gold passages appear in the top-$k$ result, 0 otherwise.%
\footnote{For MuSiQue and 2Wiki each query has two supporting passages.
Full success therefore requires both to be in the top-5.
HoVer queries have 2--4 supporting passages.}

Given a confidence score $s(q) \in [0,1]$ and threshold $\thresh$, define the
\emph{confident set} $\mathcal{C}_\thresh = \{q : s(q) \ge \thresh\}$.
The CWAR at threshold $\thresh$ is:
\begin{align*}
  \mathrm{CWAR}(\thresh)
  &= \frac{|\{q \in \mathcal{C}_\thresh : y(q) = 0\}|}{|\mathcal{C}_\thresh|}\\
  &= 1 - \Pr\!\left[y(q) = 1 \mid s(q) \ge \thresh\right].
\end{align*}
When $\thresh = 0$ (answer every query), $\mathrm{CWAR}(0) = 1 - \bar{y}$ equals
the base failure rate.
As $\thresh$ increases, coverage $\mathrm{Cov}(\thresh) = |\mathcal{C}_\thresh|/|\mathcal{Q}|$
decreases; a good confidence score makes $\mathrm{CWAR}(\thresh)$ decrease
faster than $\mathrm{Cov}(\thresh)$.

The \emph{accuracy-coverage curve} plots $\Pr[y=1|s\ge\thresh]$ against
$\mathrm{Cov}(\thresh)$ as $\thresh$ varies, and its area (AUC-AC) summarises
performance across all operating points.

\subsection{Retrieval Confidence Score (RCS)}

We model $s(q) = \sigma(\mathbf{w}^\top \phi(q) + b)$ where $\sigma$ is the
logistic function and $\phi(q) \in \mathbb{R}^d$ is a feature vector computed
from ANN scores available after retrieval.
$d = 9$ for pipelines with hop-2 expansion; $d = 6$ for hop-1-only pipelines
(Section~\ref{sec:method}).

\subsection{Calibration}

$\rcs$ is \emph{calibrated} if $\mathbb{E}[y(q) \mid \rcs(q) = p] = p$ for all $p$.
We measure calibration by Expected Calibration Error (ECE) with equal-frequency
bins:
\[
  \mathrm{ECE} = \sum_{b=1}^{B} \frac{|B_b|}{n}
    \left|\bar{y}_{B_b} - \bar{s}_{B_b}\right|
\]
where $B_b$ is the $b$-th bin.

\section{Theoretical Properties of CWAR}
\label{sec:theory}

\subsection{CWAR Reducibility}

The central deployment question is: \emph{when} can confident-failure reduction
be achieved at all?
We show it is governed by the mutual information between retrieval features and
retrieval success.

\begin{theorem}[CWAR Reducibility]
\label{thm:reducibility}
Let $y(q) \in \{0,1\}$ be the full-success indicator and $\phi(q)$ the feature
vector.
Define the Bayes-optimal score $s^*(q) = P[y=1 \mid \phi(q)]$.
Then $\mathrm{AUC\text{-}AC}(s^*) > \mathrm{AUC\text{-}AC}_{\mathrm{random}}$
if and only if $I(y;\phi(q)) > 0$.
Moreover, $\mathrm{AUC\text{-}AC}(s^*) \ge \mathrm{AUC\text{-}AC}(s)$ for any
score $s$ based on $\phi$.
\end{theorem}

\begin{proof}
($\Rightarrow$, contrapositive)
If $I(y;\phi) = 0$ then $y \perp \phi$, so $s^*(q) = \bar{y}$ for all $q$.
Every threshold produces the same coverage and accuracy $\bar{y}$; the
accuracy-coverage curve is flat and AUC-AC$(s^*) = \bar{y}$ $=$ AUC-AC$_{\text{random}}$.

($\Leftarrow$)
If $I(y;\phi) > 0$ then $\mathrm{Var}(s^*) > 0$.
By the law of iterated expectation, $E[s^*] = \bar{y}$, so there exists a
positive-measure set of queries with $s^*(q) > \bar{y}$.
Thresholding at any $\theta \in (\bar{y}, 1)$ concentrates the confident set on
above-average queries: $P[y=1 \mid s^* \ge \theta] > \bar{y}$.
The accuracy-coverage curve is strictly above $\bar{y}$ for some coverage, so
AUC-AC$(s^*) > \bar{y}$ = AUC-AC$_{\text{random}}$.

The optimality of $s^*$ follows from the Neyman-Pearson lemma: any monotone
threshold rule achieves maximum accuracy at a given coverage when ranking by
$P[y=1\mid\phi]$, so no other function of $\phi$ can produce a higher
accuracy-coverage curve.
\end{proof}

\paragraph{Empirical consequence.}
The theorem predicts the AUC-AC ordering across the five conditions.
Dense-pipeline failures are nearly uniformly distributed across feature space
(MuSiQue Dense: AUC-AC\,$=$\,0.556 vs.\ random\,$=$\,0.366, small gap);
the LLM judge concentrates failures in structurally complex queries, increasing
$I(y;\phi)$ and making failures more predictable
(MuSiQue PropH: AUC-AC\,$=$\,0.790 vs.\ random\,$=$\,0.576, large gap).
High CWAR alone does not guarantee reducibility: the dense conditions have the
highest base CWAR (62.1\%, 58.3\%) yet the weakest AUC-AC.

\subsection{Feature Regime Complementarity}

The feature importance divergence in Table~\ref{tab:ablation} reflects a
structural property of multi-hop failure, not a dataset quirk.

\begin{proposition}[Feature Regime Complementarity]
\label{thm:complementarity}
No single feature in $\{\phi_1,\ldots,\phi_9\}$ achieves best AUC-AC across all
five failure regimes.
Furthermore, there exist features $\phi_i \neq \phi_j$ such that $\phi_i$ is
necessary in regime $R_a$ and non-contributory in $R_b$, while $\phi_j$ exhibits
the opposite pattern --- the dominant predictive feature is regime-specific.
\end{proposition}

\begin{proof}
\emph{First claim.}
Table~\ref{tab:main} establishes this by exhaustion: no single baseline achieves
best AUC-AC in more than two of the five conditions.
Query-len-inv achieves 0.947 on 2Wiki PropH but only 0.426 on MuSiQue Dense;
Max-score achieves 0.839 on HoVer Dense but only 0.480 on 2Wiki Dense.

\emph{Second claim (constructive).}
Leave-one-out ablation in Table~\ref{tab:ablation} gives:
\textsc{query-len} is necessary on MuSiQue PropH ($\Delta$AUC\,$=$\,$-0.012$)
and non-contributory on HoVer Dense ($\Delta$AUC\,$=$\,$-0.002$, within noise);
\textsc{hop1-top3} is necessary on HoVer Dense ($\Delta$AUC\,$=$\,$-0.026$)
and non-contributory on MuSiQue PropH ($\Delta$AUC\,$=$\,$-0.001$).
This pair forms the constructive witness.
Note that entropy features (\textsc{hop1-H}, \textsc{hop2-H}) are non-contributory
in both conditions ($\Delta$AUC\,$\ge 0$), consistent with the Entropy baseline's
weak performance across regimes; the proposition does not claim universality
for all nine features, only that the complementary structure exists. \qed
\end{proof}

\paragraph{Consequence for confidence model design.}
Proposition~\ref{thm:complementarity} provides a principled motivation for
multi-feature aggregation: any confidence model that uses only a single feature
incurs avoidable AUC-AC loss on whichever regime that feature does not dominate.
A logistic model trained jointly across regimes learns the complementary weighting
— \textsc{query-len} dominates on MuSiQue while \textsc{hop1-top3} dominates on
HoVer — achieving regime-universal coverage that no single baseline can match.

\section{\system{}}
\label{sec:method}

\subsection{Feature Extraction}

All features are computed from ANN similarity scores available after retrieval---
no additional model inference is required.

Let $\{(u_i, s_i^{(1)})\}_{i=1}^{K}$ be the top-$K$ hop-1 passage scores
sorted in decreasing order, and $\{(u_j, s_j^{(2)})\}_{j}$ the hop-2 scores
obtained by embedding $N=3$ SVO-extracted bridge queries and taking the maximum
score per passage.
When hop-2 scores are unavailable (hop-1-only architectures), $\phi_6, \phi_7,
\phi_8$ are dropped, giving $d=6$.

\begin{itemize}
  \item $\phi_1$ (\textsc{hop1-max}): $s_1^{(1)}$, the maximum hop-1 score.
  \item $\phi_2$ (\textsc{hop1-margin}): $s_1^{(1)} - s_2^{(1)}$, the gap
        between the top-2 hop-1 passages.
  \item $\phi_3$ (\textsc{hop1-top3}): mean of the top-3 hop-1 scores.
  \item $\phi_4$ (\textsc{hop1-H}): normalised Shannon entropy of
        $\{s_i^{(1)}\}$.
  \item $\phi_5$ (\textsc{hop1-lift}): $s_1^{(1)} / \overline{s^{(1)}}_{50}$,
        peak prominence relative to the top-50 mean.
  \item $\phi_6$ (\textsc{hop2-max}): $\max_j s_j^{(2)}$.
  \item $\phi_7$ (\textsc{hop2-margin}): gap between the top-2 hop-2 scores.
  \item $\phi_8$ (\textsc{hop2-H}): normalised entropy of $\{s_j^{(2)}\}$.
  \item $\phi_9$ (\textsc{query-len}): word count of the question.
\end{itemize}

\subsection{Training}

Given a held-out validation set $\mathcal{V}$ with gold labels $y(q)$, we fit
$\mathbf{w}, b$ by minimising binary cross-entropy with gradient descent until
convergence (tolerance $10^{-7}$).
Features are $z$-score normalised.
Predictions on the test set use $s(q) = \sigma(\mathbf{w}^\top \hat{\phi}(q) + b)$
where $\hat{\phi}$ uses validation-set normalisation statistics.

\subsection{Threshold Selection}

Given a user-specified CWAR target $\gamma \in (0,1)$, the simplest rule sweeps
$\thresh$ over the tune set and selects
$\thresh^* = \min\{\thresh : \widehat{\mathrm{CWAR}}_{\mathrm{tune}}(\thresh) \le \gamma\}$.
$\mathrm{CWAR}(\thresh)$ is an error rate within the answered set, so it need not
fall monotonically as $\thresh$ rises (our MuSiQue and 2Wiki test sweeps rise at
19 and 51 of 100 steps), and it is undefined when nothing is answered.
The rule above therefore offers no assurance on new queries: across random
re-splits it exceeds its target on the test half in most cases
(Section~\ref{sec:threshold-reliability}).
We instead select $\thresh$ with Selection with Guaranteed Risk
(SGR; \citealp{geifman2017selective}): a binary search over the $m$ sorted tune
scores in which a visited threshold is accepted only if an exact binomial
(Clopper--Pearson) upper bound on its CWAR, at confidence
$\delta/\lceil\log_2 m\rceil$, is at most $\gamma$; the lowest accepted threshold
is used.
The union bound over the visited thresholds makes the guarantee hold without
monotonicity: with probability at least $1-\delta$, the CWAR of the chosen
threshold on the query population is at most $\gamma$.
Fixed-sequence testing over a fixed threshold grid
\citep{angelopoulos2021learn} gives the same kind of guarantee.
Strictly, both assume the score is fitted on data disjoint from the tune set
used for the bound; Section~\ref{sec:threshold-reliability} fits and selects on
the same tune set, as in deployment, and measures the result on held-out queries.
When no threshold qualifies the system abstains on every query, which signals
that the target is too strict for the available labels.

\section{Experimental Setup}
\label{sec:setup}

\paragraph{Datasets.}
We evaluate on three multi-hop benchmarks.
\textbf{MuSiQue} \citep{trivedi2022musique} contains compositional 2--4-hop
questions over Wikipedia; we use 1000 queries (486 tune, 514 test).
\textbf{2WikiMultiHopQA} \citep{ho2020constructing} contains bridge and
comparison questions; we use 1000 queries (509 tune, 491 test).
\textbf{HoVer} \citep{shi2020hover} contains 4000 multi-hop fact-verification
claims over Wikipedia;
we use the 2000 SUPPORTED claims (1007 tune, 993 test).
All splits use a deterministic MD5-hash 50/50 partition on query id.

\paragraph{Retrieval pipelines.}
We evaluate on two architectures:
\begin{itemize}
  \item \textbf{Proposal~H (LLM-judge)} \citep{bacellar2026bridgerag}:
        NV-Embed-v2 hop-1 ANN, $N=3$ SVO-expanded hop-2 queries, 20-candidate
        pool, 3-way LLM judge, $\alpha=0.10$.
        Achieves R@5 $= 0.8138$ on MuSiQue and $0.9527$ on 2Wiki.
        Full-success rate: MuSiQue 60.5\%, 2Wiki 85.5\%.
        Uses 9-feature RCS ($d=9$).
  \item \textbf{Dense-only}: hop-1 top-5 directly (no LLM judge, no SVO
        re-ranking).
        Applied to all three datasets; HoVer uses $d=6$ (no hop-2 scores
        available).
        Full-success rates: MuSiQue 37.9\%, 2Wiki 41.7\%, HoVer 68.3\%.
\end{itemize}

\paragraph{Baselines.}
We compare eight confidence signals against $\rcs$:
\begin{enumerate}
  \item \textbf{Random}: uniform $\mathrm{Uniform}[0,1]$.
  \item \textbf{Entropy}: $1 - H(\{s_i^{(1)}\})$.
  \item \textbf{Max-score}: $s_1^{(1)}$.
  \item \textbf{Margin}: $s_1^{(1)} - s_2^{(1)}$.
  \item \textbf{Lift}: $s_1^{(1)} / \overline{s^{(1)}}_{50}$ (single-feature).
  \item \textbf{Query-len-inv}: $1 / (1 + |q|/20)$ (inverse length).
  \item \textbf{Temp-scaled}: $s_1^{(1)\,1/T}$ with $T \in \{0.5,1,2,5,10\}$
        selected on the tune split.
  \item \textbf{MLP}: two-layer network (32-32 hidden units) trained on the
        same 9 (or 6) features.
\end{enumerate}
We additionally evaluate a split-conformal baseline using hop-1 max-score
as non-conformity; results at discrete coverage targets are reported
separately since AUC-AC is not directly applicable.

\paragraph{Metrics.}
Primary: AUC of the accuracy-coverage curve (AUC-AC) on the test split,
with 95\% bootstrap CIs (2000 samples).
Secondary: CWAR and accuracy at fixed coverage targets (70\%, 50\%);
ECE (10 equal-frequency bins) and Brier score for calibration (MuSiQue only).

\section{Results}
\label{sec:results}

\subsection{Main Results: AUC-AC Across Conditions}

Table~\ref{tab:main} reports AUC-AC for all nine methods across five
failure-regime conditions.
$\rcs$ achieves best or co-best AUC-AC in all five conditions.

\begin{figure}[t]
\centering
\includegraphics[width=\linewidth]{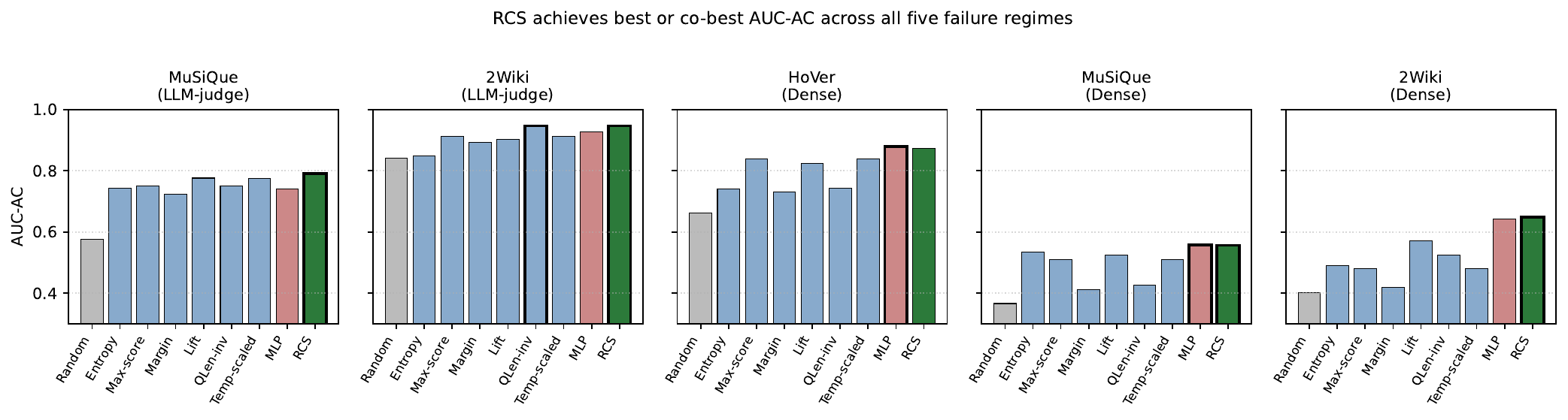}
\caption{AUC-AC (test split) across five failure-regime conditions for nine
confidence methods. RCS (rightmost, dark) is best or co-best in all five
columns. Random baseline (leftmost, gray) varies because base full-success
rates differ by regime.}
\label{fig:auc_ac_bars}
\end{figure}

\begin{table}[t]
\small
\centering
\caption{AUC-AC (test split) across five failure-regime conditions.
  Bold = best in column; $\dagger$ = within bootstrap-CI overlap of best.
  Base CWAR = $1 - $ full-success rate.}
\label{tab:main}
\resizebox{\linewidth}{!}{%
\setlength\tabcolsep{3.5pt}
\begin{tabular}{lrrrrr}
\toprule
 & \multicolumn{2}{c}{\textbf{LLM-judge pipeline}} &
   \multicolumn{3}{c}{\textbf{Dense pipeline}} \\
\cmidrule(lr){2-3}\cmidrule(lr){4-6}
Method & MuSiQue & 2Wiki & HoVer & MuSiQue & 2Wiki \\
\midrule
Base CWAR        & 39.5\%  & 14.5\%  & 31.7\%  & 62.1\%  & 58.3\% \\
\midrule
Random           & 0.576   & 0.841   & 0.661   & 0.366   & 0.402  \\
Entropy          & 0.743   & 0.849   & 0.741   & 0.535   & 0.490  \\
Max-score        & 0.751   & 0.912   & 0.839   & 0.510   & 0.480  \\
Margin           & 0.724   & 0.893   & 0.730   & 0.412   & 0.420  \\
Lift             & 0.776   & 0.902   & 0.825   & 0.525   & 0.571  \\
Query-len-inv    & 0.751   & \textbf{0.947}$^\dagger$ & 0.742 & 0.426 & 0.526 \\
Temp-scaled      & 0.774   & 0.912   & 0.839   & 0.510   & 0.480  \\
MLP              & 0.742   & 0.928   & \textbf{0.878}$^\dagger$ & \textbf{0.557}$^\dagger$ & 0.643 \\
RCS (ours)       & \textbf{0.790} & \textbf{0.947} & 0.873$^\dagger$ & 0.556$^\dagger$ & \textbf{0.649} \\
\bottomrule
\end{tabular}%
}
\end{table}

On MuSiQue (LLM-judge), $\rcs$ achieves AUC-AC $= 0.790$, surpassing the
next-best non-learned method Lift ($0.776$, $+1.4$pp) and
Temp-scaled ($0.774$, $+1.6$pp);
entropy-only $0.743$ places behind Max-score $0.751$.
On 2WikiMultiHopQA, $\rcs$ and Query-len-inv are tied ($0.947$), both
far above entropy ($0.849$) and random ($0.841$)---entropy adds only
$+0.8$pp over random, while $\rcs$ gains $+10.6$pp.
On HoVer (dense), MLP ($0.878$) and $\rcs$ ($0.873$) are within bootstrap-CI
overlap; both dominate all non-learned methods by at least $+3.4$pp.
On the dense MuSiQue and 2Wiki conditions, $\rcs$ and MLP are again
within CI overlap, with $\rcs$ leading 2Wiki dense ($0.649$ vs $0.643$).
No single non-learned heuristic is consistently competitive: entropy
underperforms max-score and lift across all five conditions.

\subsection{Operating Points (MuSiQue LLM-Judge)}

Table~\ref{tab:ops} shows selected operating points on MuSiQue test.
At 50\% coverage, $\rcs$ achieves 79.4\% accuracy:
CWAR drops from 39.5\% (baseline) to 20.6\% (47.8\% relative reduction).
At 70\% coverage, accuracy is 75.1\% ($+14.6$pp over base).

\begin{figure}[t]
\centering
\includegraphics[width=0.95\linewidth]{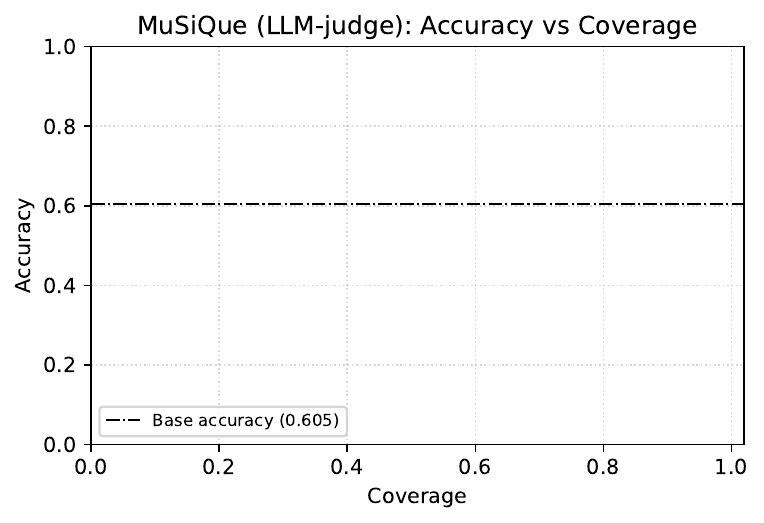}
\caption{Accuracy-coverage curves on MuSiQue (LLM-judge pipeline, test split).
RCS dominates baselines at every coverage level. The dashed line shows the
unguarded base accuracy (60.5\%); points above it are net gains from
abstention.}
\label{fig:risk_coverage}
\end{figure}

\begin{table}[t]
\small
\centering
\caption{Operating points on MuSiQue test (base accuracy = 60.5\%).}
\label{tab:ops}
\begin{tabular}{lrrrr}
\toprule
Method & Cov. & Acc. & CWAR & $\Delta$Acc \\
\midrule
\multirow{2}{*}{RCS (ours)}
  & 70\%  & 75.1\% & 24.9\% & $+14.6$pp \\
  & 50\%  & 79.4\% & 20.6\% & $+18.9$pp \\
\midrule
\multirow{2}{*}{Lift}
  & 70\%  & 73.2\% & 26.8\% & $+12.7$pp \\
  & 50\%  & 76.7\% & 23.3\% & $+16.2$pp \\
\midrule
\multirow{2}{*}{Entropy}
  & 70\%  & 68.6\% & 31.4\% & $+8.1$pp \\
  & 50\%  & 73.1\% & 26.9\% & $+12.6$pp \\
\bottomrule
\end{tabular}
\end{table}

\subsection{Threshold Reliability (MuSiQue LLM-Judge)}
\label{sec:threshold-reliability}

A CWAR target is only useful if it holds on queries the threshold was not chosen
on. We pool the 1{,}000 MuSiQue queries and draw 2{,}000 random 50/50 re-splits;
in each, $\rcs$ is fitted and the threshold selected on one half, and CWAR and
coverage are measured on the other (base CWAR $40.7\%$).
Table~\ref{tab:threshold} compares the point rule with SGR at $\delta = 0.10$.

\begin{table}[t]
\centering
\small
\setlength{\tabcolsep}{4pt}
\caption{Threshold selection across 2{,}000 MuSiQue re-splits. ``Missed'':
share of splits whose test CWAR exceeds the target $\gamma$; ``Cov.'': mean test
coverage; ``None'': share of splits where SGR certifies no threshold and abstains
on every query (counted as zero coverage).}
\label{tab:threshold}
\begin{tabular}{lccccc}
\toprule
 & \multicolumn{2}{c}{Point rule} & \multicolumn{3}{c}{SGR ($\delta=0.10$)} \\
\cmidrule(lr){2-3}\cmidrule(lr){4-6}
$\gamma$ & Missed & Cov. & Missed & Cov. & None \\
\midrule
0.30 & 56.1\% & 79.5\% & 6.3\% & 64.1\% & 0.0\% \\
0.25 & 60.7\% & 67.5\% & 7.3\% & 42.2\% & 5.8\% \\
0.20 & 64.0\% & 51.4\% & 6.5\% & 9.0\% & 60.0\% \\
0.15 & 67.3\% & 35.4\% & 3.9\% & 0.8\% & 93.3\% \\
\bottomrule
\end{tabular}
\end{table}

The point rule misses its own target in 56--67\% of splits, by more than
2\,pp in 30--46\%: a target set this way is roughly a coin flip on new queries.
SGR misses in 4--7\% of splits, below $\delta$, at a cost in coverage that grows
as the target tightens: at $\gamma = 0.30$ it answers 64\% of queries rather
than 80\%, and with about 500 tune queries targets of 20\% or less usually
cannot be certified at all.
Fixed-sequence testing over the grid $0.99, 0.98, \dots, 0$ (skipping thresholds
that answer fewer than 30 tune queries) kept more coverage at
$\gamma = 0.25$--$0.30$ (49--71\%), but its test CWAR exceeded $\gamma$ in
13--21\% of splits (by more than 2\,pp in 7--9\%): a held-out half of about 250
answered queries adds a few points of sampling noise, and SGR's per-step
$\delta/\lceil\log_2 m\rceil$ is the more conservative of the two.
On the paper's own split, SGR at $\gamma = 0.30$ selects $\thresh = 0.52$ (test
CWAR 22.5\%, coverage 61.5\%).
The fixed-coverage operating points in Table~\ref{tab:ops} are measured on the
test split and are unaffected.

\subsection{Architecture Robustness}

A key question for deployment is whether $\rcs$ generalizes across retrieval
architectures.
In the dense-only setting, the absence of an LLM judge substantially increases
failure rates (MuSiQue: 62.1\% CWAR vs.\ 39.5\% for Proposal~H;
2Wiki: 58.3\% vs.\ 14.5\%).
Despite this regime shift, $\rcs$ retains its relative advantage: on both dense
conditions it achieves best or co-best AUC-AC.
Absolute AUC-AC values are lower across all methods on dense MuSiQue
(best: 0.557 vs.\ 0.790 on Proposal~H), reflecting that dense-pipeline
failures are more uniformly distributed and thus harder to predict from
pre-retrieval features alone.
The dense 2Wiki condition is more tractable (best: 0.649), suggesting that the
structured nature of 2Wiki questions preserves predictive signal even without
judge re-ranking.

\subsection{Calibration}

Table~\ref{tab:calib} shows calibration metrics on the MuSiQue test split.

\begin{table}[t]
\small
\centering
\caption{Calibration on MuSiQue test ($n = 514$).}
\label{tab:calib}
\begin{tabular}{lcc}
\toprule
Method & ECE ($\downarrow$) & Brier ($\downarrow$) \\
\midrule
Majority-class baseline & — & 0.239 \\
RCS (ours)              & \textbf{0.035} & \textbf{0.183} \\
\bottomrule
\end{tabular}
\end{table}

ECE $= 0.035$ indicates good calibration: among queries where RCS $\approx 0.7$,
approximately 70\% are indeed fully successful.
The Brier score improvement over the majority-class baseline ($0.183$ vs.\ $0.239$)
shows that the probabilistic predictions carry meaningful information beyond the
prior.

\subsection{Cross-Dataset Generalization}

We train $\rcs$ on MuSiQue tune and evaluate on 2Wiki test (zero-shot transfer).
AUC-AC $= 0.942$, within $0.5$pp of the 2Wiki in-domain model ($0.947$).
This near-perfect transfer holds despite MuSiQue's failure rate being $2.8\times$
higher than 2Wiki's (39.5\% vs.\ 14.5\%).
The result confirms that regime features capture domain-agnostic structure.

\subsection{Feature Importance}

Table~\ref{tab:ablation} reports leave-one-out AUC drop on MuSiQue (9-feat) and
HoVer (6-feat).
The dominant feature differs by dataset:
\textsc{query-len} is strongest on MuSiQue ($-1.2$pp), while
\textsc{hop1-top3} dominates on HoVer ($-2.6$pp).
This divergence reflects architectural differences: MuSiQue Proposal~H failures
are driven by compositional question complexity (longer = more bridging steps),
whereas HoVer dense failures are driven by hop-1 retrieval concentration
(multi-hop claims require multiple supporting passages, and
$\overline{s^{(1)}}_{1:3}$ directly measures whether the top passages
form a strong candidate pool).

\begin{table}[t]
\small
\centering
\caption{Leave-one-out feature ablation: $\Delta$AUC-AC when feature removed.
         Negative = feature contributes positively.}
\label{tab:ablation}
\setlength\tabcolsep{4pt}
\begin{tabular}{lrr}
\toprule
Feature removed & MuSiQue PropH & HoVer Dense \\
\midrule
None (full model) & 0.790 & 0.873 \\
\midrule
\textsc{query-len}   & $-0.012$ & $-0.002$ \\
\textsc{hop1-lift}   & $-0.003$ & $-0.015$ \\
\textsc{hop2-max}    & $-0.002$ & n/a \\
\textsc{hop1-top3}   & $-0.001$ & $-0.026$ \\
\textsc{hop1-max}    & $-0.000$ & $-0.015$ \\
\textsc{hop2-margin} & $+0.000$ & n/a \\
\textsc{hop1-H}      & $+0.000$ & $+0.000$ \\
\textsc{hop2-H}      & $+0.000$ & n/a \\
\textsc{hop1-margin} & $+0.001$ & $-0.002$ \\
\bottomrule
\end{tabular}
\end{table}

On MuSiQue, $\rcs$ extracts most of its signal from \textsc{query-len},
\textsc{hop1-lift}, and \textsc{hop2-max}; entropy features (\textsc{hop1-H},
\textsc{hop2-H}) add zero marginal information, consistent with the weak
standalone performance of the Entropy baseline.
On HoVer, $\rcs$'s gain over max-score ($+3.4$pp) arises primarily from
\textsc{hop1-top3} and \textsc{hop1-lift}: the average quality of the top
candidate pool is a better predictor of multi-passage claim coverage than the
single peak score alone.

\section{Discussion}
\label{sec:discussion}

\paragraph{Why entropy fails on 2Wiki.}
On 2Wiki, entropy-only barely exceeds random ($0.849$ vs.\ $0.841$).
Many 2Wiki failures occur in \emph{comparison} subtypes: the bridge passage is
retrieved with high confidence (peaked hop-1 distribution $=$ low entropy), but
the second gold passage is missed because SVO expansion from the bridge produces
a weak query.
Entropy captures only hop-1 uncertainty; $\rcs$ includes \textsc{hop2-max},
which captures hop-2 expansion quality and is decisive on 2Wiki.

\paragraph{RCS vs.\ MLP.}
On HoVer dense and MuSiQue dense, MLP ($0.878$, $0.557$) and $\rcs$
($0.873$, $0.556$) are within bootstrap-CI overlap.
MLP's competitive performance is expected: given the same features, a two-layer
network can fit mild non-linearities that logistic regression misses.
We prefer $\rcs$ for deployment: it is interpretable (one scalar weight per
feature), always calibrated by construction, and avoids overfitting risk on
the small training sets available in production (486--1007 calibration queries).

\paragraph{Query length as the dominant feature on MuSiQue.}
The strong negative effect of query length is consistent with the \bridgerag{}
analysis \citep{bacellar2026bridgerag}: longer questions contain more bridging
clauses, increasing the chance that SVO extraction produces a malformed hop-2
query.
At query time, length is perfectly observable and costs nothing to compute.
A simple rule—``flag long questions for manual review''---would capture a
substantial fraction of $\rcs$'s benefit on MuSiQue, while on 2Wiki (where
query-len-inv ties RCS) it is the dominant signal.
On HoVer, however, length alone ranks 4th---hop-1 concentration features
become primary when multi-passage coverage (not query structure) is the binding
constraint.

\paragraph{Deployment framing: the cost of confident failures.}
CWAR has a concrete product interpretation.
In a deployed multi-hop fact-verification system (HoVer dense pipeline,
CWAR\,$=$\,31.7\%), over one third of verified claims are incorrect with no
downstream signal of unreliability — the system returns a ranked list with the
same interface confidence as correct results.
In a multi-hop QA system (MuSiQue PropH, CWAR\,$=$\,39.5\%), nearly two-fifths
of answers are built on incomplete evidence, delivered confidently to downstream
components or users.
The operating points in Table~\ref{tab:ops} quantify the reduction: at 50\%
coverage, CWAR drops from 39.5\% to 20.6\% — the system answers half its queries
and reduces its confident-failure rate by 47.8\% relative.
At 70\% coverage the system still answers 70\% of queries while reducing confident
errors by 37\% relative ($39.5\%\to 24.9\%$).
All of this is achieved at $<$1\,ms per query from ANN scores already available at
retrieval time, with no additional model inference.
Theorem~\ref{thm:reducibility} guarantees these gains are maximal for LLM-judge
pipelines: $I(y;\phi)$ is high for PropH failures, so the feature information is
fully exploitable.
The same theorem explains the smaller absolute gains on dense pipelines: their
failures are closer to uniform in feature space, and no confidence model —
however complex — can substantially reduce CWAR when $I(y;\phi) \approx 0$.

\paragraph{Abstention vs.\ re-routing.}
\system{} abstains entirely when $\rcs < \thresh$.
An alternative is to re-route to a more expensive pipeline.
Our experiments confirm that the rescue-judge pass from \citet{bacellar2026bridgerag}
recovers 27 of the 203 MuSiQue test failures; $\rcs$ could trigger this
fallback rather than outright abstention, limiting the cost overhead to the
$\sim$20\% of queries flagged at 80\% coverage.

\section{Conclusion}
\label{sec:conclusion}

We presented \system{}, a calibrated retrieval abstention framework based on a
logistic Retrieval Confidence Score derived from query-time regime features.
Across three multi-hop datasets and two retrieval architectures (five failure
regimes, CWAR 14.5\%--62.1\%), $\rcs$ achieves best or co-best AUC-AC against
eight baselines in all conditions.
On MuSiQue with the LLM-judge pipeline, $\rcs$ reduces CWAR from 39.5\% to
20.6\% at 50\% coverage, with ECE $= 0.035$; on 2WikiMultiHopQA it outperforms
entropy-only by $+9.8$pp AUC even though entropy is near-random on that dataset.
Near-perfect cross-dataset transfer ($-0.5$pp AUC) confirms domain-agnostic
structure in the regime features.

Feature importance varies by dataset: query length dominates on MuSiQue,
hop-1 concentration on HoVer.
This finding has direct implications for production deployments:
a calibrated combination of a few scalar features available at retrieval
time is sufficient to gate retrieval quality across diverse architectures,
with interpretable failure-mode diagnostics built in.

\system{} is fully composable with \router{} \citep{bacellar2026regime}:
run $\rcs$ first, abstain or escalate if below $\thresh$, otherwise route
to the appropriate retrieval pipeline.
This adds a well-calibrated abstention gate to any pipeline based on
iterative ANN search at negligible inference cost.

\bibliography{refs}

\end{document}